\documentclass[10pt,draftcls,onecolumn]{IEEEtran}

\usepackage{flushend}
\usepackage{graphicx}
\usepackage{amsthm}
\usepackage{amsmath}
\usepackage{amssymb}
\usepackage{mathrsfs}
\usepackage{cite}
\usepackage{hyperref}
\usepackage{url}
\usepackage{multirow}
\usepackage{rotating}
\usepackage{color}
\usepackage{graphicx}
\usepackage{subfigure}
\usepackage{bbm}
\usepackage{algpseudocode}

\DeclareMathOperator{\X}{\mathbf{X}}

\def\x{{\mathbf x}}

\newtheorem{prop}{Proposition}
\newtheorem{mydef}{Definition}
\newtheorem{theorem}{Theorem}

\newtheorem{remark}[theorem]{Remark}

\title{A Coalitional Game for Distributed Inference in Sensor Networks with Dependent Observations}

\author{
Hao He and Pramod K. Varshney \\
Department of EECS, Syracuse University, Syracuse, NY 13244, USA
}

\begin{document}
\maketitle

\begin{abstract}
We consider the problem of collaborative inference in a sensor network with heterogeneous and statistically dependent sensor observations. Each sensor aims to maximize its inference performance by forming a coalition with other sensors and sharing information within the coalition. 
It is proved that the inference performance is a nondecreasing function of the coalition size. However, in an energy constrained network, the energy consumption of inter-sensor communication also increases with increasing coalition size, which discourages the formation of the grand coalition (the set of all sensors). 
In this paper, the formation of non-overlapping coalitions with statistically dependent sensors is investigated under a specific communication constraint. We apply a game theoretical approach to fully explore and utilize the information contained in the spatial dependence among sensors to maximize individual sensor performance. Before formulating the distributed inference problem as a coalition formation game, we first quantify the gain and loss in forming a coalition by introducing the concepts of \textit{diversity gain} and \textit{redundancy loss} for both estimation and detection problems. These definitions, enabled by the statistical theory of copulas, allow us to characterize the influence of statistical dependence among sensor observations on inference performance.  An iterative algorithm based on merge-and-split operations is proposed for the solution and the stability of the proposed algorithm is analyzed. Numerical results are provided to demonstrate the superiority of our proposed game theoretical approach.  
\end{abstract}


\section{introduction}
In a distributed inference problem, each sensor collects observations regarding a phenomenon of interest, then shares them with other sensors or transmits them to the fusion center (FC). To reduce the energy cost for communication, the observations may be processed before transmission. The distributed nature of wireless sensor networks induces a tradeoff between minimizing the communication cost and maintaining acceptable performance levels. Although there has been a lot of work on distributed inference, including distributed detection and distributed estimation, with conditionally independent observations, much less has been done for the case of dependent observations \cite{goodbadugly,vvv06,cheng09,Niu2010,Sundaresan11,Sundaresan11a,Iyengar11,Iyengar12,HaoChen12,Shen2014}.

The spatial correlation among the sensor observations is a significant characteristic which
can be exploited to significantly enhance the overall network performance, including inference performance and energy efficiency. Typical applications of Wireless Sensor Networks (WSNs) require spatially dense sensor deployment in order to achieve satisfactory coverage. As a result, proximal sensors recording information about a single event are highly correlated with the degree of correlation increasing with decreasing internode separation. Such dependence among adjacent sensors or agents also exists in other intelligence aggregation networks. For example, in a crowd sourcing network, agents with the same backgrounds or having active interactions (e.g., following each other on social websites) are likely to have correlated knowledge/observations about the same event. Any network consisting of dependent agents having the ability to take measurement of the environment and making inference based on available observations,  such as wireless sensor networks, cognitive radio networks or a crowd sourcing network, are within the consideration of this work. For the simplicity of presentation, we use the term ``sensor" to represent an intelligent agent, which can be a real sensor, a cognitive radio, or a participating agent in a crowd sourcing network, in the remainder of this paper. Dependence among observations may make some sensors' observations redundant. An extreme case is when two sensors' observations are completely positively correlated, one of the two sensor will become ``redundant". Since transmitting ``redundant" observations from battery powered sensors to remotely located FC is energy inefficient, we have an opportunity to conserve energy via local collaboration in a densely located sensor network. 

The effect of dependent noise and hence dependent observations on Fisher Information (FI) has been studied by Yoon and Sompolinsky in \cite{Yoon98theeffect}. The authors showed that, in the biologically relevant regime of parameters, positive correlations degrade estimation performance compared with an uncorrelated population. Sundaresan et al.~\cite{Sundaresan11} considered location estimation of a random signal source where they focused on improving system performance by exploiting the spatial dependence of sensor observations. Parameter estimation with dependent observations in a variety of communication scenarios was considered in~\cite{ SHson2005}, but was limited to the case of ``geometric'' dependent Gaussian noise.  

Different approaches have been employed to study the detection problem with correlated observations, most of which focus on small sample size \cite{goodbadugly,HaoChen12}. It has been shown that correlation degrades overall performance either in the case of a binary signal in equicorrelated Gaussian noise \cite{vvv03} or in the cases where correlation increases with the decrease in inter-sensor distance \cite{vvv07}. In parallel sensor networks, the fusion of  statistically dependent observations is considered under various scenarios \cite{Sundaresan11a,Iyengar11,Iyengar12,hao12,hao14} and the design of local decision rules is investigated in \cite{HaoChen12} through the introduction of hierarchical independence model. Performance
of WSNs exposed to correlated observations is also assessed using the theory of large deviations \cite{vvv06}.

In this paper, we formulate a novel distributed inference framework which fully exploits and utilizes the inter-sensor dependence for improved overall system performance, given the inherent tradeoff between inference performance and transmission efficiency. This framework provides an opportunity to maintain a comparable inference performance to that of a centralized framework while achieving greater transmission efficiency, in networks with correlated sensors. In such a framework, there is no FC and each individual sensor is capable of sensing and computing. Sensors form non-overlapping coalitions and collaborate by sharing their observations within a coalition. In the process of forming coalitions, each sensor selfishly aims to maximize its own inference performance, and thus the performance of the coalition to which it belongs, as will be evident later. The problem is to find a set of non-overlapping coalitions such that each sensor's inference performance is maximized under certain energy cost constraints. To model and analyze the spatial dependence among sensor observations which might be heterogeneous (different marginal distributions), we use copula theory, which has been applied for inference with dependent observations in \cite{Iyengar11,Iyengar12,Sundaresan11a ,Sundaresan11}.

In our framework, each sensor is characterized not only by its individual inference performance achieved with its own observations, but also by its dependence with other sensors in the network. Unlike the individual performance which is fixed and unchangeable no matter which coalition the sensor belongs to, its dependence with other sensors plays different roles in different coalitions. In order to quantify the gain and loss of collaboration resulting from inter- sensor dependence, we introduced the concepts of redundancy loss and diversity gain for the distributed estimation in \cite{HaoHe13}. Other definitions of \textit{diversity} are available in different contexts in the signal processing literature. In cognitive radio systems, \textit{diversity} is acknowledged as the benefit of collaborative sensing and diversity order in various collaborative spectrum sensing schemes is quantitatively determined in \cite{cooperativediv}. In communication systems, \textit{diversity} is widely adopted as an indicator of the signal-to-noise ratio (SNR) dependent behavior of inference performance based on multiple received signals \cite{tse2005fundamentals,estimationdiv}. 
In the distributed inference problem that we are considering in this paper, diversity gain quantifies the positive effect of dependence on inference performance, in contrast with redundancy loss, which quantifies the redundant information induced by the dependence among sensor data. 

Since an optimal solution to the coalition formation problem may not exist, namely there may not be such a partition that every sensor's performance is maximize, our best hope is to find a stable solution. Thus, we use game theoretical approach and formulate our collaborative distributed inference problem as a coalition formation game. Game theory has been widely applied to statistical inference, such as measurement allocation for localization \cite{Farhad2008}, communication networks \cite{Walid2009}, and spectrum sensing \cite{Zhu2009}.  An iterative algorithm based on merge and split operations \cite{Apt_ageneric} is proposed in the literature to find a stable solution for the coalitional games discussed above.

Building on our preliminary work in \cite{HaoHe13}, which focused on an estimation problem, in this paper, we study the general problem of distributed inference in sensor networks with local collaboration. The major contributions of this paper can be summarized as follows:
\begin{itemize}
\item We fully investigate the different roles played by inter-sensor dependence for inference problems, including both detection and estimation; and we define diversity gain and redundancy loss to respectively characterize the benefit and loss in forming coalitions due to inter-sensor dependence in this more generalized setting.
\item We formulate a coalition formation game for the more generalized distributed inference problem with dependent observations for large heterogeneous sensor networks. We design an iterative algorithm based on merge and split operations to solve the coalition formation game, which is more efficient than other approaches available in the existing literature. 
\end{itemize}

The rest of the paper is organized as follows. Basic concepts of copula theory and coalitional games are introduced in Section \ref{sec:prelim} as background knowledge. Section \ref{sec:sm} describes the system model and the inference problem is formulated in general. Section \ref{sec:cde} introduces the distributed estimation problem and
analyzes the role of inter-sensor dependence. Section \ref{sec:dd} analyzes the  problem of distributed detection and quantifies the dependence-related diversity gain and redundancy loss.
Section \ref{sec:gf} proposes a coalition formation game and a merge-and-split based algorithm to obtain a stable solution. Section \ref{sec:sim} presents and discusses simulation results. We provide concluding remarks in Section \ref{sec:con}.

\section{Preliminaries}
\label{sec:prelim}
\subsection{Copula Theory}

Simply put, copula functions \emph{couple} multivariate joint distribution functions to their component marginal distribution functions \cite{Nelsen2006}. We begin with the definition of a copula function.

\begin{mydef}
\label{def:Copula}
A function $C:[0,1]^N \to [0,1]$ is an N-dimensional copula if $C$ is a joint cumulative distribution function (CDF) of an N-dimensional random vector on the unit cube $[0,1]^N$ with uniform marginals \cite{Joe1997, Kurowicka2006, Nelsen2006}.
\end{mydef}
The application of copulas to statistical signal processing is made possible largely because of the following theorem by Sklar \cite{Nelsen2006}.
\begin{theorem}
[Sklar's Thoerem] 
Consider an $N$-dimensional distribution function $F$ with marginal distribution functions $F_1,\ldots,F_N$. Then there exists a copula $C$, such that for all $x_1,\ldots,x_N$ in $[-\infty,\infty]$
\begin{equation}
\label{CopEq1}
F(x_1,x_2,\ldots,x_N) = C(F_1(x_1),F_2(x_2),\ldots,F_N(x_N))
\end{equation}
If $F_n$ is continuous for $1\leq n \leq N$, then $C$ is unique.
\label{thm:sklar}
\end{theorem}
Conversely, given a copula $C$ and univariate CDFs $F_1,\ldots,F_N$, $F$ as defined in \eqref{CopEq1} is a
valid multivariate CDF with marginals $F_1,\ldots,F_N$.
 According to Sklar's Theorem~\cite{Nelsen2006}, for continuous distributions, the joint probability density function (PDF) can be obtained by differentiating both sides of \eqref{CopEq1}
\begin{equation}
\label{CopEq2}
 f(x_1,\ldots,x_N) = \left(\prod_{n=1}^{N}f_n(x_n)\right)c(F_1(x_1),\ldots,F_N(x_N)|\boldsymbol{\phi})
\end{equation}
where $c(\cdot)$ is termed as the copula density function and is given by  
\begin{equation}
\label{CopDens}
 c(\mathbf{u}) = \frac{\partial^NC(u_1,\ldots,u_N)}{\partial u_1,\ldots,\partial u_N}
\end{equation}
with $u_n=F_n(x_n)$. Copula functions contain a \emph{dependence parameter} $\boldsymbol{\phi}$ that quantifies the amount of dependence among the $N$ random variables. It needs to be noted that this is well-suited for modeling heterogeneous random vectors where a different distribution might be needed to model each marginal $x_n$. Several copula functions are defined in the literature~\cite{Nelsen2006} of which the elliptical and Archimedean copulas are widely used.  

An attractive feature of copulas is their relationship with the nonparametric rank-based measures of dependence, such as Kendall's $\tau$
\footnote{ Let $(X_1,Y_1)$ and $(X_2,Y_2)$ be two independent pairs of random variables with a common joint distribution function $H$ and copula $C$, The population version $\tau_{X,Y}$ of Kendall's $\tau$ is defined as the probability of concordance minus the probability of discordance: $\tau_{X,Y} = P[(X_1-X_2)(Y_1-Y_2) >0]-P[(X_1-X_2)(Y_1-Y_2)<0])$.}.
The relationship for a copula $C$, and the Kendall's $\tau$ for random variables $X$ and $Y$ is given by ~\cite[p. 159]{Nelsen2006} 
\begin{equation}
\tau_{X,Y}=4\int \int C(u,v) \mathrm{d} C(u,v) -1
\label{rhocopula}
\end{equation}
where $u=F_X(x),v=F_Y(y)$. The relationship in \eqref{rhocopula} results in a one-to-one correspondence between Kendall's $\tau$ and copula parameter $\boldsymbol{\phi}$, based on which a rank-based estimation of dependence parameter $\boldsymbol{\phi}$ can be performed. 

\subsection{Coalitional Game Theory}
To facilitate the formulation of our problem, we introduce basic concepts in coalitional game theory. Let $\mathcal{N}=\{1, 2, \dots, N\}$ be a set of fixed players called the \textit{grand coalition}. Nonempty subsets of $\mathcal{N}$ are called \textit{coalitions}. A \textit{collection} (in the grand coalition $\mathcal{N}$) is any family $\mathcal{S} :=\{S_1, \dots, S_m\}$ of mutually disjoint coalitions. If additionally $\cup_{j=1}^m S_j=\mathcal{N}$, the collection $\mathcal{S}$ is called a \textit{partition} of $\mathcal{N}$. 

Assuming a comparison relation $\triangleright$, $\mathcal{R}=\{R_1, \dots, R_k\} \triangleright \mathcal{S}=\{S_1, \dots, S_m\}$ means that the way $\mathcal{R}$ partitions $\mathcal{N}$, where $\mathcal{N}= \cup_{i=1}^k R_i= \cup_{j=1}^m S_j$, is preferred over the way $\mathcal{S}$ partitions $\mathcal{N}$ based on some performance measure.
Pareto order can be used as a comparison relation $\triangleright$.
For a collection $\mathcal{R}=\{R_1, \dots, R_k\}$, the utility of a player $j$ in a coalition $R_j \in \mathcal{R}$ is denoted by $\Phi_j(\mathcal{R})$,
and the Pareto order is defined as follows
\begin{equation}
\mathcal{R} \triangleright \mathcal{S} \iff \{\Phi_j(\mathcal{R}) \geq \Phi_j(\mathcal{S}), \forall j \in \mathcal{R},\mathcal{S} \}
\end{equation}
with at least one strict inequality for a player $k$. 

Apt and Witzel~\cite{Apt_ageneric} proposed an abstract approach to coalition formation that focuses on simple merge-and-split rules to transform partitions of a group of players. Details of coalition formation will be introduced in detail in Section \ref{sec:cde}.

\section{system model}
\label{sec:sm}
We consider a physical phenomenon being continuously observed by a set of densely deployed sensors,  which is represented by $\mathcal{N}=\{1,2, \dots, N\}$. Each sensor's observation is $x_n$. Let
$\theta$ be the parameter that denotes the phenomenon of interest in the received signal $x_n$ at sensor $n$ for the general inference problem.
When we consider a detection problem, $\theta$ represents a binary discrete variable, while in the case of parameter estimation, $\theta$ is a realization of a continuous random variable $\Theta$ with PDF $f_{\Theta}(\cdot)$. Due to high density of sensors in the network topology, sensor observations are highly correlated spatially. 

%

In a non-collaborative setting, each sensor continuously senses the environment, and locally makes inference about the unknown parameter $\theta$ solely based on its own observations. In this work, we consider a collaborative setting where collaboration exists within coalitions. Participating sensors are required to act in accordance with the following rules: 
\begin{enumerate}
\item Sensors first form coalitions, and each sensor can only join one coalition.  
\item Once the coalitions are formed, a sensor can request observations from all the other sensors in the same coalition and make an inference; it also has to transmit its observations to the other collaborating sensors upon their request.
\end{enumerate}
In such a collaborative setting, each sensor, as an independent agent, aims to improve its own inference performance through collaboration with the most ``useful" sensors. The coalition formation process, namely, how the coalitions should be formed such that each selfish sensor has its performance maximized, is the focus in this paper. 

An intuitive solution would be that all the sensors form a grand coalition such that every sensor enjoys the benefit of collaboration to the maximum extent. However, in an energy constrained network, each sensor's energy is finite and a communication cost is incurred when it transmits its observations to collaborating sensors. Let $r$ be the average number of requests initiated by each sensor in the network per unit time interval. Then, for any sensor in coalition $S$, the number of requests that have to be responded to within a unit time interval is $r (|S|-1)$, where $|S|$ denotes the cardinality of coalition $S$. 
We assume that energy consumption for a single transmission is $E_t$. The average energy consumption per unit time interval for each sensor in coalition $S$ is $E(S)=r (|S|-1)E_t$, which increases as the coalition size increases. Let the energy consumption of a coalition be the average energy consumption per sensor in this coalition, which is the same quantity $E(S)=r (|S|-1)E_t$. Thus, from the point of view of energy consumption, smaller coalitions are preferred. In order to guarantee adequate sensors' lifetime, we enforce the energy consumption constraint as follows
\begin{eqnarray}
E(S) =r (|S|-1)E_t < \alpha, \quad \forall S \in \mathcal{S}.
\label{eq:energycon}
\end{eqnarray}

Then the problem is to find the optimal partition $\mathcal{S}$ of the set of sensors $\mathcal{N}$ such that each sensor's inference performance is maximized subject to the energy constraint in (\ref{eq:energycon}).  
\begin{eqnarray}
&&\max_{\mathcal{S} \in \mathcal{P}} \Delta_n(\mathcal{S}), \quad \forall n \in \mathcal{N} \nonumber \\
&& \text{subject to} ~E(S)< \alpha, \quad \forall S \in \mathcal{S}
\label{eq:optimization}
\end{eqnarray}
where $\Delta_n(\mathcal{S})$ represents the inference performance of sensor $n$ under partition $\mathcal{S}$, and $\mathcal{P}$ denotes the set of all possible partitions of $\mathcal{N}$.

For the optimization problem in (\ref{eq:optimization}), an exhaustive approach in which we search over all possible partitions will invoke a very high computational complexity. According to \cite{Sandholm1999209}, for a network with $N$ sensors, the total number of partitions is $O(N^N)$.  
Besides computational issues, an exhaustive search may not be able to give us a solution to the problem in (\ref{eq:optimization}), since there may not exist a partition such that each sensor's performance is maximized simultaneously while the energy consumption constraint is satisfied. For the same reason, if each sensor solves its optimization problem iteratively by itself, the overall system optimization algorithm may not converge.
Thus, our best hope is to find a stable solution \footnote{A stable solution can simply be interpreted as a partition where no player has the incentive to leave the current partition. Stability will be discussed in detail later in this work.} and to do that we use a game theoretical approach. Before formulating the distributed inference problem as a coalition formation game, we need to define and analyze the gain and the loss of each sensor when it joins a coalition, in the context of dependent observations. The analysis is carried out respectively for the problem of estimation and detection in the following two sections.

\section{Collaborative Distributed Estimation}
\label{sec:cde}
In the estimation problem, the optimization problem can be formulated as the minimization of Posterior Cramer-Rao Lower Bound (PCRLB), or equivalently, the maximization of posterior Fisher Information (FI), which  is given by
\begin{eqnarray}
&&FI =-\mathbb{E}_{\X, \Theta} \left[ \frac{\partial^2}{\partial \theta^2} \log f_{\X}(\x;\theta) \right]\nonumber \\
&& = -\mathbb{E}_{\X, \Theta}\left[ \frac{\partial^2}{\partial \theta^2} \log f_{\X}(\x;\theta) \right]  \nonumber \\
&&- \mathbb{E}_{ \Theta}\left[ \frac{\partial^2}{\partial \theta^2} \log f_{\Theta}(\theta)\right] \nonumber \\
&& = I +I_P
\end{eqnarray}
where $f_{\X}$ represents the joint PDF of $\X:=[X_1,\dots, X_N]$; $I$ and $I_P$ represent the sensor data's contribution and prior's contribution to posterior FI respectively.
The prior's contribution is fixed given the distribution of $\Theta$. Thus, we only consider sensor data's contribution. Since $I$ is the FI averaged over the distribution of $\Theta$, it is referred to as the average FI \cite{avgfi}. 
For the coalition $S$ whose set of observations is $\x_S :=[x_n, \forall n \in S]$, the average FI it can achieve is given as  
\begin{eqnarray}
\label{eq:FIofS}
I(S)=-\mathbb{E} \left[ \frac{\partial^2 \log f_{\X_S}(\x_S;\theta)}{\partial \theta^2}\right]
\end{eqnarray}
where $f_{\X_S}(\cdot)$ denotes the joint distribution of $\X_S$ and the expectation is taken with respect to the joint distribution of $\X_S, \Theta$.

\begin{remark}
\label{remark:FI}
As an immediate result of the modus operandi of the network, the estimation performance, i.e., average FI, achievable at sensor $n$ that is in coalition $S$, denoted by $I_n(S)$, equals the average FI contained in coalition $S$, which is denoted by $I(S)$. That is 
\[
I_n(S) = I(S), ~ \forall n \in S
\].
\end{remark}

\begin{prop}
$I(S)$ is a nondecreasing function of the cardinality of $S$.
\label{prop:nonde}
\end{prop}

\begin{proof} 
We need to show that $I(S) \geq I(S^{'})$, for $S^{'} \subseteq S$. According to the definition of average FI of coalition $S$ in (\ref{eq:FIofS})
\begin{eqnarray}
I(S)&=&-\mathbb{E}\left[\frac{\partial^2}{\partial \theta^2} \log f_{\X_S}(\x_S;\theta)\right] \nonumber \\
&=& -\mathbb{E}_{S^{'}}\left[\frac{\partial^2}{\partial \theta^2} \log f_{\X_{S^{'}}}(\x_{S^{'}};\theta) \right]+ \nonumber \\
&&\mathbb{E}_{{S^{'}}}\left[-\mathbb{E}_{{S\setminus S^{'}}|{S^{'}}} [\frac{\partial^2}{\partial \theta^2}\log f_{\X_{S\setminus S^{'}}}(\x_{S\setminus S^{'}}|\x_{S^{'}};\theta)]\right]  \nonumber \\
\label{decoupleIS}
\end{eqnarray}
where $S\setminus S^{'}$ denotes the relative complement of $S^{'}$ with respect to $S$, i.e., $\{n: n \in S, n \notin S^{'}\}$. It can be noted that the first term in \eqref{decoupleIS} corresponds to the average FI of $S^{'}$, and the second term is the expected conditional average FI of $S \setminus S^{'}$. Due to the non-negativity of conditional FI, we have
\begin{eqnarray}
I(S)&=& I(S^{'})+\mathbb{E}_{{S^{'}}}\left[ I(S\setminus S^{'}|{S^{'}})\right] \nonumber \\
&\geq&  I(S^{'}) 
\end{eqnarray}
\end{proof}
\begin{remark}
When the transmission cost is assumed to be zero, i.e., $E_t=0$, a grand coalition forms. It is proved in Proposition \ref{prop:nonde} that average FI does not decrease by including more sensors in a coalition. Thus, if there is no communication cost, all the sensors will collaborate for a better estimation performance. 
\end{remark}

It is clear from Proposition \ref{prop:nonde} and the definition of $E(S)=r (|S|-1)E_t$ that, as the coalition size increases, both estimation performance in terms of $I(S)$ and the energy consumption $E(S)$ increase with it. There is a tradeoff between the estimation performance and communication efficiency. Each sensor aims to maximize its estimation performance subject to an energy constraint. The problem is formulated as the following:
\begin{eqnarray}
&&\max_{\mathcal{S} \in \mathcal{P}} I_n(\mathcal{S}), \quad \forall n \in \mathcal{N} \nonumber \\
&&\text{s.t.} \quad r (|S|-1)E_t < \alpha, \quad \forall S \in \mathcal{S}
\end{eqnarray}
where $I_n(\mathcal{S})$ represents the average FI of sensor $n$ under partition $\mathcal{S}$, i.e. $I_n(\mathcal{S})=I_n(S)$, for $n \in S$ and $S \in \mathcal{S}$. 
%

\subsection{Diversity gain \& Redundancy loss }

%
To analyze the effect of inter-sensor dependence on the average FI for coalition $S$, we express the joint PDF of observations of sensors in coalition $S$ in terms of the marginal PDFs and copula density function $c_s$, as in \eqref{CopEq2}, using copula theory.
When $\log c_S(\cdot;\theta,\boldsymbol{\phi})$ is twice differentiable with respect to $\theta$, $I(S)$ can be written as 
\begin{eqnarray}
I(S)&=&-\mathbb{E} \left[  \frac{\partial^2  \log \left(  \prod_{n \in S}f_{n}(x_{n}; \theta) c_S(\cdot ;\theta,\boldsymbol{\phi})\right)}{\partial \theta^2} \right] \nonumber\\
&=&\sum_{n\in S} {I_n} - \mathbb{E}\left[\frac{ \partial^2\log c_S(\cdot ;\theta, \boldsymbol{\phi})} {\partial \theta^2}\right] \nonumber\\
&=&\sum_{n\in S} {I_n}+I_{c}(S)
\label{eq:fc}
\end{eqnarray}
where $I_n$ represents the average FI achieved by a single sensor $n$ in a non-collaborative setting, and  $I_{c}(S)$ represents the FI that is induced by the dependence structure $c_S$. Thus, the average FI for a coalition $S$ can be written as the summation of average FIs of each individual sensors in $S$ and $I_{c}(S)$. We call $I_{c}(S)$ the generalized average FI (GAFI) for the copula density function $c_S$ because it may not satisfy the non-negativity property of average FI. Figure \ref{fig:GAFI_RHO} shows the GAFI for a Gaussian copula as a function of the dependence parameter $\rho$. It is shown that for the case of identical marginal distributions, GAFI is nonpositive and decreases with an increase in $\rho$. More complicated behavior of GAFI is observed when marginal distributions are different as seen in Figure \ref{fig:GAFI_RHO}. 

\begin{figure}
\centering
\includegraphics[width=0.5\columnwidth]{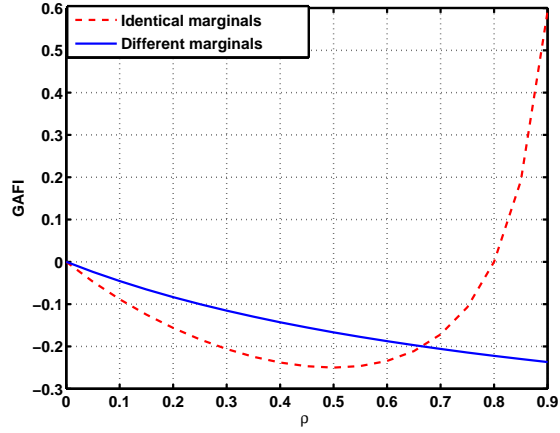}
\caption{GAFI corresponding to Gaussian copula vs. correlation coefficient $\rho$. (The marginal distributions are Gaussian. Identical marginals imply that the marginal distributions are the same, and heterogeneous marginals imply that the marginal distributions are different.)}
\label{fig:GAFI_RHO}
\end{figure}

The following proposition provides some insights into the properties of the GAFI for a two-sensor coalition. We assume the joint distribution to be bivariate Gaussian which can be written as a product of Gaussian marginals and a Gaussian copula.
\begin{prop}
Let the random vector $[X,Y]^T$ be bivariate Gaussian distributed, i.e., $[X,Y]^T \sim N(\boldsymbol{\mu}, \Sigma_{XY})$, where $\boldsymbol{\mu}=[\mu_X(\theta), \mu_Y(\theta)]^T$,
\[
\Sigma_{XY}=
\begin{pmatrix} 
\sigma_X^2 & \sigma_X \sigma_Y \rho_{XY} \\ 
\sigma_Y \sigma_X \rho_{YX} & \sigma_Y^2 
\end{pmatrix}  
\]
and $\theta$ is the parameter to be estimated (Without loss of generality, let $\left|{\frac{\sigma_X}{\sigma_Y}\frac{\mu_Y^{'}(\theta)}{\mu_X^{'}(\theta)} }\right| \leq 1$, where the derivatives are taken with respect to $\theta$), then we have: 
\begin{enumerate}
\item $I_{c}(X,Y) $, the GAFI of copula $c_{XY}$, is a convex function of $\rho_{XY}$ and $\min_{\rho_{XY}} I_{c}(X,Y)=-\frac{\mu_Y^{'2} (\theta)}{\sigma_Y^2}$ is reached at $\rho_{XY}=\frac{\sigma_X}{\sigma_Y}\frac{\mu_Y^{'}(\theta)}{\mu_X^{'}(\theta)}$; 
\item $I_{c}(X,Y) \leq 0$ when $\rho_{XY} $ lies between $0$ and 
$\frac{2 \mu_X^{'}(\theta) \mu_Y^{'}(\theta) \sigma_X \sigma_Y}{\mu_X^{'2} (\theta) \sigma_Y^2+ \mu_Y^{'2}(\theta) \sigma_X^2}$. 
\item When ${\frac{\sigma_X}{\sigma_Y}\frac{\mu_Y^{'}(\theta)}{\mu_X^{'}(\theta)} }=1$, $I_{c}(X,Y)  \geq 0$ for $\rho \in [-1,0]$ and $I_{c}(X,Y) < 0$ for $\rho \in (0,1]$. Furthermore, $I_{c}(X,Y) $ is a monotonically decreasing function of $\rho_{XY}$.
\end{enumerate}
\label{prop:bifc}
\end{prop}
\begin{proof}
According to the definition of GAFI in (\ref{eq:fc})
\begin{eqnarray}
I_{c}(X,Y)&=& - \mathbb{E}\left[\frac{ \partial^2\log c_{XY}(F_X(x;\theta),F_Y(y;\theta) ;\rho_{XY})} {\partial \theta^2}\right] \nonumber\\
&=&\frac{-1}{\sigma_X^2 \sigma_Y^2 (1-\rho_{XY}^2)} \{ 2 \rho_{XY} \mu_X^{'}(\theta) \mu_Y^{'}(\theta) \sigma_X \sigma_Y  \nonumber \\
&&-\rho_{XY}^2 (\mu_X^{'2} (\theta) \sigma_Y^2+\mu_Y^{'2}(\theta) \sigma_X^2)\}
\label{eq:icxy}
\end{eqnarray}
It can be shown that \footnote{The dependence of $I_{c}(X,Y)$ on correlation coefficient $\rho_{XY}$ is not made explicit for notational convenience.}
\[
\frac{\partial^2I_{c}(X,Y)}{\partial \rho_{XY}^2} \geq 0, \quad \forall \rho_{XY} \in (-1,1)
\]
thus, $I_{c}(X,Y)$ is convex. By setting 
\[
\frac{\partial I_{c}(X,Y)}{\partial \rho_{XY}}=0
\]
and knowing that 
\[
\left|{\frac{\sigma_X}{\sigma_Y}\frac{\mu_Y^{'}(\theta)}{\mu_X^{'}(\theta)} }\right| \leq 1
\]
we get 
 \[
 \rho^*=\frac{\sigma_X}{\sigma_Y}\frac{\mu_Y^{'}(\theta)}{\mu_X^{'}(\theta)}, \quad I_{c}(X,Y)^*=-\frac{\mu_Y^{'2} (\theta)}{\sigma_Y^2}
\] 
Thus, the minimum of $I_{c}(X,Y)$ is obtained at $\rho^*=\frac{\sigma_X}{\sigma_Y}\frac{\mu_Y^{'}(\theta)}{\mu_X^{'}(\theta)}$, which is $I_{c}(X,Y)^*=-\frac{\mu_Y^{'2} (\theta)}{\sigma_Y^2}$.


By setting (\ref{eq:icxy}) equal to zero, we get two solutions: 
\[
\rho_1=0, \quad \rho_2=\frac{2 \mu_X^{'}(\theta) \mu_Y^{'}(\theta) \sigma_X \sigma_Y}{\mu_X^{'2} (\theta) \sigma_Y^2+ {\mu_Y (\theta)^{'}}^2 \sigma_X^2}
\]
Combined with the convexity of the function, it can be concluded that  $I_{c}(X,Y) \leq 0$ when $\rho_{XY} \in \left[ \min \{ \rho_1, \rho_2\}, \max \{\rho_1, \rho_2\} \right]$. 

By letting $\sigma_X \mu_Y^{'}(\theta)=\sigma_Y \mu_X^{'}(\theta)$ in (\ref{eq:icxy}), the conclusions in (3) can be directly derived.  
\end{proof}

\begin{remark}
When $\rho_{XY}=0$, $I_{c}(X,Y)=0$, meaning that the average FI of the coalition is solely the summation of individual average FIs of $X$ and $Y$; when $\rho_{XY}=\frac{\sigma_X}{\sigma_Y}\frac{\mu_Y^{'}(\theta)}{\mu_X^{'}(\theta)}$, $I_{c}(X,Y)$ is just the smaller individual average FI of the two sensors with a minus sign. In the latter case, the sensor with larger individual average FI gains nothing in estimation performance by collaboration. 
\end{remark}
\begin{remark}
In our formulation, a sensor $n$ prefers to collaborate with sensor $m$ with which it has a positive $I_{c}(X_n,X_m)$ than sensor $k$ with which it has a negative $I_{c}(X_n,X_k)$, when sensor $m$ and sensor $k$ have identical individual performances in terms of average FI. This is because to sensor $n$, sensor $m$ is more ``valuable''  than sensor $k$ in the sense that due to inter-sensor dependence, some of sensor $k$'s information is redundant for sensor $n$.

\end{remark}

\begin{mydef}
For $I_{c}(X,Y) < 0$, we define -$I_{c}(X,Y)$ to be pairwise redundancy loss denoted as $I_{rl}(X,Y)$,  otherwise we define $I_{c}(X,Y)$ to be pairwise diversity gain denoted as $I_{dg}(X,Y)$.
\end{mydef}

The definitions of diversity gain and redundancy loss allow for a better characterization of the different roles that pairwise inter-sensor dependence may play.  General properties of the GAFI of multivariate copulas can be  analyzed using \textit{vines} which is a graphical method of constructing multivariate copulas \cite{Cooke2006,Subramanian2011}. The joint PDF of $N$ random variables expressed in terms of a D-vine decomposition is given by: 
\begin{eqnarray}
\lefteqn{f_{\X}(\x)=} \nonumber \\
&\prod \limits_{n=1}^N f(x_n)\prod \limits_{j=1}^{N-1} \prod \limits_{k=1}^{N-j} c_{j,j+k|\overline{j}}(F(x_j|\x_{\overline{j}}),F(x_{j+k}|\x_{\overline{j}}))
 \label{eq:vine}
\end{eqnarray}
where $\overline{j}=[j+1, \dots, j+k-1]$ and $\x_{\overline{j}}=[x_{j+1}, \dots, x_{j+k-1}]$.
Thus, a multivariate copula is decomposed into the product of bivariate conditional copulas. Therefore, $I_{c}(S)$, the corresponding GAFI of the copula in any coalition $S$ can be written as:
\begin{eqnarray}
&I_{c}(S)&=\sum_{j=1}^{|S|-1} \sum_{k=1}^{|S|-j} I_{c}(X_j,X_{j+k}|\X_{\overline{j}}) \nonumber \\
&=&\sum_{j=1}^{|S|-1} \sum_{k=1}^{|S|-j} I_{dg}(X_j,X_{j+k}|\X_{\overline{j}}) \mathbbm{1}_{I_{c}(X_j,X_{j+k}|\X_{\overline{j}}) \geq 0 } \nonumber\\
 &&- \sum_{j=1}^{|S|-1} \sum_{k=1}^{|S|-j} I_{rl}(X_{j},X_{j+k}|\X_{\overline{j}}) \mathbbm{1}_{I_{c}(X_j,X_{j+k}|\X_{\overline{j}})<0}  \nonumber\\
 &=& I_{dg}(S)-I_{rl}(S)
\label{eq:decfc}
\end{eqnarray}
where $\mathbb{I}_{\{\cdot\}}$ denotes the indicator function, and $I_{dg}(S)$ and $I_{rl}(S)$ respectively represent the diversity gain and redundancy loss in the coalition $S$. Each of them is a summation of pairwise diversity gains or pairwise redundancy losses in coalition $S$.  Until now, we have quantified the benefit and cost of forming a coalition $S$ incurred by dependent sensor observations in the problem of distributed estimation.
 In the following section, the counterparts of diversity gain and redundancy loss for the distributed detection problem will be investigated. 

\section{Collaborative distributed detection}
\label{sec:dd}
In the detection problem, $\theta$ is a bi-valued variable which takes the value $\theta_0$ under hypothesis $H_0$ and takes the value $\theta_1$ under hypothesis $H_1$. In this paper, we employ Kullback-Leibler Divergence (KLD) as the performance metric. KLD can be interpreted as the error exponent in the Neyman-Pearson framework, which means that the probability of miss detection goes to zero exponentially with the number of observations at a rate equal to KLD. Thus, KLD characterizes the asymptotic detection performance. We denote KLD by $D$ and define it as follows 
\begin{eqnarray}
D=\mathbb{E}_{H_0}\left[ \log \frac{f_{\X}(\x|H_0)}{f_{\X}(\x|H_1)}\right]
\end{eqnarray}
where $\mathbb{E}_{H_0}[\cdot]$ denotes the expectation taken with respect to the joint distribution of $\X$ under hypothesis $H_0$. For a coalition $S$, the detection performance it can achieve, in terms of KLD, is $D(S)$. 
\begin{eqnarray}
D(S)=\mathbb{E}_{X_S|H_0}\left[ \log \frac{f_{\X_S}(\x_S|H_0)}{f_{\X_S}(\x_S|H_1)}\right]
\end{eqnarray}
\begin{remark}
The KLD for sensor $n$, i.e., $D_n(S)$, in a coalition $S$ is the same for all $n \in S$. Similar to Remark \ref{remark:FI} for the estimation problem, we can write
\[
D_n(S) = D(S), ~\forall n \in S
\].
\label{remark:D}
\end{remark}
\begin{prop}
D(S) is nondecreasing in $|S|$.
\label{prop:kldnonde}
\end{prop}
\begin{proof}
In order to prove that $D(S)$ does not decrease by including new members to the existing coalition, we need to show that for any $S^{'} \subseteq S$, $D(S^{'}) \leq D(S)$.
\begin{eqnarray}
&&D(S) = \int \log \frac{f_{\X_S}(\x_S|H_0)}{f_{\X_S}(\x_S|H_1)} f_{\X_S}(\x_S|H_0) \mathrm{d}\x_S \nonumber \\
&& = \int \log \frac{f_{\X_{S^{'}}}(\x_{S^{'}}|H_0)}{f_{\X_{S^{'}}}(\x_{S^{'}}|H_1)} f_{\X_{S^{'}}}(\x_{S^{'}}|H_0)  \mathrm{d}\x_{S^{'}} \nonumber \\
&& + \int \log \frac{f_{\X_{S\setminus S^{'}}}(\x_{S\setminus S^{'}}|\x_{S^{'}},H_0) }{f_{\X_{S\setminus S^{'}}}(\x_{S\setminus S^{'}}|\x_{S^{'}},H_1) } f_{\X_S}(\x_S|H_0) \mathrm{d}\x_{S} \nonumber \\
&& = D(S^{'})+\mathbb{E}_{\X_{S^{'}}|H_0}\left[D(S\setminus S^{'}) \right] \nonumber \\
&& \geq D(S^{'})
\end{eqnarray}
The last inequality is because of the non-negativity property of conditional KLD.
\end{proof}

\begin{remark}
A grand coalition forms when communication cost is zero, i.e., $E_t=0$. 
\end{remark}

It is noted that, as $|S|$ increases, both $D(S)$ and $E(S)$ increase, indicating a tradeoff between the detection performance and energy consumption. In our formulation, each sensor selfishly aims to maximize its own detection performance, i.e., the KLD that can be obtained using shared observations within the coalition to which it belongs, subject to an energy constraint. The problem can be formulated as the following
\begin{eqnarray}
&&\max_{\mathcal{S} \in \mathcal{P}} D_n(\mathcal{S}), \quad \forall n \in \mathcal{N} \nonumber \\
&&\text{s.t.} \quad r (|S|-1)E_t < \alpha, \quad \forall S \in \mathcal{S}
\end{eqnarray}
where $D_n(\mathcal{S})$ represents the KLD of sensor $n$ under partition $\mathcal{S}$, i.e., $D_n(\mathcal{S})=D_n(S)$, for $S \in \mathcal{S}$ and $n \in S$ . 
\subsection{Diversity Gain and Redundancy Loss}
The effect of inter-sensor dependence on the KLD can be analyzed by expressing the joint PDF of the observations of sensors in coalition $S$ in terms of the marginal PDFs and copula density function $c_s$.
By copula theory, the KLD corresponding to $\X_S$ can be written as 
\begin{eqnarray}
\lefteqn{D(S) =} \nonumber \\
&&\int \log \frac{ \prod \limits_{n \in S}f_{n}(x_{n}|H_0) c_S(\cdot |\boldsymbol{\phi}_0,H_0)}{\prod \limits_{n \in S}f_{n}(x_{n}|H_1) c_S(\cdot |\boldsymbol{\phi}_1,H_1)} f_{\X_S}(\x_S|H_0) \mathrm{d}\x_S \nonumber \\
&& = \sum_{n \in S} D_n +\mathbb{E}_{\X_S|H_0}\left[ \log \frac{c_S\left(F_{n}(x_{n}|H_0), \forall n \in S |\boldsymbol{\phi}_0,H_0\right)}{c_S\left(F_{n}(x_{n}|H_1), \forall n \in S |\boldsymbol{\phi}_1,H_1\right)} \right]\nonumber \\
&& = \sum_{n \in S} D_n+ D_{c}(S)
\label{decomd}
\end{eqnarray}
where $D_n$ is the KLD achieved by sensor $n$ with its own observations in a non-collaborative setting and $\boldsymbol{\phi}_i$ is the dependence parameter of the copula density under hypothesis $H_i$, $i = 0,1$. 
The KLD between the two joint distributions of sensor observations in $S$ under hypotheses $H_0$ and $H_1$ can be decomposed into two terms, as shown in \eqref{decomd}. The first term represents the summation of KLDs corresponding to individual sensors in $S$ and the second term $D_{c}(S)$ measures the \textit{distance} between the two joint distributions introduced by the dependence structure. We call $D_{c}(S)$ the Generalized KLD (GKLD), because the arguments of $c_S(\cdot|H_0)$ and $c_S(\cdot|H_1)$ are different and thus violate the standard definition of KLD. In Figure \ref{fig:GKLD_TAU}, the GKLDs corresponding to different copulas are plotted against Kendall's $\tau$. Similar trend is observed among these curves. For each curve, there exist a single $\tau^*$ that divides $\tau \in [0,1]$ into two intervals, each corresponding to positive or negative GKLD. Since Kendall's $\tau$ is only a scalar summarization of the ``amount" of dependence, the behaviors of GKLDs vary for different copula models (structures of the dependence). 

\begin{figure}
\centering
\includegraphics[width=0.5\columnwidth]{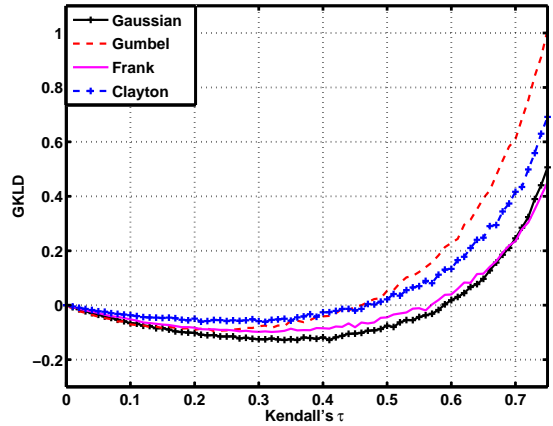}
\caption{GKLD corresponding to different copulas vs. Kendall's $\tau$: Gaussian marginals are assumed.}
\label{fig:GKLD_TAU}
\end{figure}

The following proposition provides insights into the GKLD in a coalition consisting of two sensors whose observations follow bivariate Gaussian distribution which can be viewed as the product of two univariate Gaussian PDFs and a Gaussian copula.
\begin{prop}
Consider two random variables $[X,Y]^T \sim N([\theta_1, \theta_1], \Sigma_{XY})$ under hypothesis $H_1$ and $[X,Y]^T \sim N([\theta_0, \theta_0], \Sigma_{XY})$ under hypothesis $H_0$, where
\[
\Sigma_{XY}=
\begin{pmatrix} 
\sigma_X^2 & \sigma_X \sigma_Y \rho_{XY} \\ 
\sigma_Y \sigma_X \rho_{YX} & \sigma_Y^2 
\end{pmatrix}  
\]
and $\theta_1 \neq \theta_0$. Without loss of generality, let $\sigma_X \leq \sigma_Y$, then we have: 
\begin{enumerate}
\item $D_{c}(X,Y) $, the GKLD corresponding to the Gaussian copula $c_{XY}$, is a convex function of $\rho_{XY}$ and $\min_{\rho_{XY}} D_{c}(X,Y)=-\frac{(\theta_1-\theta_0)^2}{2 \sigma_Y^2}$ is reached at $\rho_{XY}=\frac{\sigma_X}{\sigma_Y}$; \\
\item $D_{c}(X,Y) \leq 0$ for $\rho_{XY} $ between $0$ and $\frac{2\sigma_X \sigma_Y}{ \sigma_Y^2+ \sigma_X^2}$. 
\item  For $\sigma_X=\sigma_Y$, $D_{c}(X,Y) \geq 0$ for $\rho_{XY} \in [-1,0]$ and  $D_{c}(X,Y) < 0$ for $\rho_{XY} \in (0,1]$ and it is a monotone decreasing function of $\rho_{XY}$,.
\end{enumerate}
\label{prop:bigkld}
\end{prop}
\begin{proof}
According to the definition of GKLD, we have
\begin{eqnarray}
\lefteqn{D_{c}(X,Y)} \nonumber \\
&&=\mathbb{E}_{XY|H_0}\left[ \log \frac{c\left(F_{X}(x|H_0),F_{Y}(y|H_0) |\Sigma_{XY}\right)}{c\left(F_{X}(x|H_1),F_{Y}(y|H_1) |\Sigma_{XY}\right)} \right] \nonumber \\
&&=\frac{(\theta_1-\theta_0)^2}{2 \sigma_X^2 \sigma_Y^2 (1-\rho_{XY}^2)} \left[ \rho_{XY}^2 ( \sigma_Y^2+\sigma_X^2)- 2 \rho_{XY}  \sigma_X \sigma_Y  \right]
\label{eq:fcb}
\end{eqnarray}
where $c(\cdot|\Sigma_{XY})$ represents the Gaussian copula parameterized by $\boldsymbol{\phi} =\Sigma_{XY}$.
It can be shown that 
\[
\frac{\partial^2D_{c}(X,Y)}{\partial \rho_{XY}^2} \geq 0,  \quad \forall \rho_{XY} \in (-1,1)
\]
Thus, the convexity is proved. By setting 
\[
\frac{\partial D_{c}(X,Y)}{\partial \rho_{XY}}=0
\]
and knowing that $\sigma_X \leq \sigma_Y$, we get 
\[
\rho^*=\frac{\sigma_X}{\sigma_Y}, \quad D_{c}(X,Y)^*=-\frac{(\theta_1-\theta_0)^2}{2 \sigma_Y^2}
\]
combining with the convexity of the function, we know that $D_{c}(X,Y)^*$ is the minimum point. 

By setting (\ref{eq:fcb}) equal to zero, we get two solutions: 
\[
\rho_1=0, \quad \rho_2=\frac{2 \sigma_X \sigma_Y}{\sigma_Y^2+  \sigma_X^2} 
\]
knowing the convexity of the function, it can be concluded that  $D_{c}(X,Y) \leq 0$ when $\rho_{XY} \in \left[ \rho_1, \rho_2 \right]$.  

When $\sigma_X=\sigma_Y=\sigma$, 
\[
D_{c}(X,Y)=-\frac{(\theta_1-\theta_0)^2}{\sigma^2} \frac{\rho_{XY}}{1+\rho_{XY}}
\]
It can be shown that $\frac{\partial D_{c}(X,Y)}{\partial \rho_{XY}} \leq 0, \forall \rho_{XY} \in (-1,1)$ and the sign of $D_{c}(X,Y)$ is the same as that of $-\rho_{XY}$.
\end{proof}
\begin{remark}
When $X$ and $Y$ are independently distributed, i.e., $\rho_{XY}=0$, then $D_{c}(X,Y)=0$, meaning that KLD is solely the summation of individual KLDs of $X$ and $Y$; when $\rho_{XY}=\frac{\sigma_X}{\sigma_Y}$, $D_{c}(X,Y)$ is just the smaller individual KLD of the two with a minus sign. In the latter case, the sensor with larger KLD does not improve its detection performance by forming a coalition with the other sensor. 
\end{remark}
\begin{remark}
For three sensors having dependent observations, a sensor $n$ would prefer to collaborate with sensor $m$ with which it has a positive $D_{c}(X_n, X_m)$ than sensor $k$ with which it has a negative $D_{c}(X_n, X_k)$, when sensor $m$ and $k$ have identical individual performance, i.e., $D_m=D_k$. This is because to sensor $n$, sensor $m$ is more ``valuable'' than sensor $k$ in the sense that the dependence between sensor $n$ and $m$ results in a larger total KLD, and thus contributes to a better asymptotic detection performance. 
\end{remark}

\begin{mydef}
For $D_{c}(X,Y) < 0$, we define -$D_{c}(X,Y)$ to be pairwise redundancy loss of GKLD, denoted as $D_{rl}(X,Y)$,  otherwise we define $D_{c}(X,Y)$ to be pairwise diversity gain of GKLD denoted as $D_{dg}(X,Y)$.
\end{mydef}

Although the expressions of the pairwise redundancy loss and diversity gain depend on the specific problem that we are considering, these definitions capture the intrinsic characteristics of a sensor network with dependent observations and quantify the impact of the dependence in collaboration.

According to (\ref{eq:vine}), a multivariate copula is decomposed into the product of bivariate conditional copulas. Therefore, $D_{c}(S)$, the GKLD introduced by the copula in any coalition $S$ can be written as:
\begin{eqnarray}
&&D_{c}(S)=\sum_{j=1}^{|S|-1} \sum_{k=1}^{|S|-j} D_{c}(X_j,X_{j+k}|\X_{\overline{j}}) \nonumber \\
&&=\sum_{j=1}^{|S|-1} \sum_{k=1}^{|S|-j} D_{dg}(X_j,X_{j+k}|\X_{\overline{j}}) \mathbbm{1}_{D_{c}(X_j,X_{j+k}|\X_{\overline{j}}) \geq 0} \nonumber\\
 &&- \sum_{j=1}^{|S|-1} \sum_{k=1}^{|S|-j} D_{rl}(X_j,X_{j+k}|\X_{\overline{j}}) \mathbbm{1}_{D_{c}(X_j,X_{j+k}|\X_{\overline{j}})<0}  \nonumber\\
 &&= D_{dg}(S)-D_{rl}(S)
\label{eq:decfc}
\end{eqnarray}
$D_{dg}(S)$ represents the diversity gain in the coalition $S$ and $D_{rl}(S)$ represents the amount of redundant information included in coalition $S$. By noting that $D_{dg}(S)$ and $D_{rl}(S)$ are nonnegative and nondecreasing function of $|S|$, we can view $D_{dg}(S)$ together with $\sum_{n\in S} D_n $ as the gain of forming $S$, while $D_{rl}(S)$ as the cost, along with the communication cost $E(S)$. In the following section, a coalition formation game for distributed inference is formulated based on the quantification of dependence-based diversity gain and redundancy loss.

\section{Game Formulation and Properties}
\label{sec:gf}
We propose a coalitional game defined by the pair $(\mathcal{N}, V)$ to model our collaborative inference problem, where $\mathcal{N}$ is the set of players (all sensors) and $V$ is a mapping such that for every coalition $S$, $V(S)$ is a closed convex subset of $\mathbb{R}^{S}$ that contains the payoffs that players in $S$ can achieve. In order to present a generalized game theoretical approach to the distributed inference problem, we use a unified notation $\Delta$ to represent the average FI in the estimation problem and the KLD in the detection problem. We define the value of a coalition $v(S)$, as an increasing function of the gain $\sum_{n \in S}\Delta_n+\Delta_{dg}(S)$ and a decreasing function of the costs $\Delta_{rl}(S)$, and $E(S)$:
\begin{eqnarray}  
\label{eq:uf}
v(S)=\left[\sum_{n \in S}\Delta_n+\Delta_{dg}(S)\right]-\left[ \Delta_{rl}(S)+C(S) \right]
\end{eqnarray}
where $C(S)$ is a function of the energy consumption $E(S)$. It captures the tradeoff between inference performance and the energy consumption. There are certain properties that a well designed cost function $C(S)$ should satisfy, here we use the logarithmic barrier penalty function given in~\cite{boyd2004}
\begin{eqnarray}
C(S)=\left\{ \begin{array}{rl}
-1/t \cdot \log(1-\frac{E(S)}{\alpha}) &  \mbox{ if $E(S) < \alpha$} \\
+ \infty & \mbox{ otherwise}
       \end{array} \right.
\label{ccf}
\end{eqnarray}
where $\alpha$ is the constraint on $E(S)$, and $t$ is a control parameter. The above cost function is an increasing function of $E(S)$ for $E(S) < \alpha$, while it goes to infinity when $E(S) \geq \alpha$. Through the cost function in \eqref{ccf}, the constraint that $E(S) < \alpha$ is enforced, since for the coalitions that do not satisfy this constraint, the utility $v(S)$ is $-\infty$.

\begin{prop}
The payoff for each sensor in coalition $S$ is equal to the utility of the coalition, i.e., $\Phi_n(S)=v(S), \forall n \in S$, where $\Phi_n(S)$ denotes the payoff of sensor $n$ in the coalition $S$. 
\end{prop}

\begin{proof}
The value of a coalition $S$ defined in (\ref{eq:uf}) is a function of its inference performance and its average energy consumption $E(S)$. According to Remarks \ref{remark:FI} and \ref{remark:D}, the average FI or KLD for every sensor in $S$ is given by the average FI and KLD of the coalition. And it is known that transmission cost $E(S)$ of every sensor in $S$ is the average transmission cost of the coalition. Hence, the coalition value $v(S)$ is also the payoff of each player in it. 
\end{proof}

Now, we have a nontransferable utility coalitional game $(\mathcal{N}, V)$, where $V(S)$ is a singleton set (hence closed and convex)
\begin{eqnarray}
V(S):= \{ \boldsymbol{\Phi}(S)| \Phi_n(S) = v(S), \forall n \in S\}
\end{eqnarray}
A distributed algorithm for the above coalitions formation game among sensors is described next. 

\subsection{Coalition formation algorithm}

For autonomous coalition formation,
we propose a distributed algorithm based on two simple
rules called \textit{merge} and \textit{split}~\cite{Apt_ageneric} that allow us to modify a
partition $\mathcal{S}$ of the set $\mathcal{N}$.

Merge Rule: Merge any set of coalitions $\{S_1, \dots, S_m\}$, where  $\{\cup_{j=1}^m S_j\} \triangleright \{S_1, \dots, S_m\}$, therefore, $\{S_1, \dots, S_m\} \to \{\cup_{j=1}^m S_j\}$.

Split Rule: Split any coalition $\{\cup_{j=1}^m S_j\}$, where $\{S_1, \dots, S_m\} \triangleright \{ \cup_{j=1}^m S_j\}$, thus $\{\cup_{j=1}^m S_j\} \to \{S_1, \dots, S_m\}$.

\begin{remark}
Every iteration of the merge and split rules terminates.
\end{remark}


Let us assume that the dependence information is known at the local sensors,  and they autonomously form coalitions through merge and split operations. Let the initial partition be $\mathcal{S}=\{ S_1, \dots, S_m\}$.

\begin{algorithmic}

\Repeat {\\
~~$\mathcal{R}~=~\text{Merge}(\mathcal{S})$: coalitions in $\mathcal{S}$ merge according to the merge rule, until no further merge occurs\\
~~$\mathcal{S}~=~\text{Split}(\mathcal{R})$: coalitions in $\mathcal{R}$ split according to the split rule, until no further split occurs.

}
\Until{No merge or split occurs}
\end{algorithmic}

Merge operations are first applied. Given an initial partition $\mathcal{S}= \{S_1, \dots, S_m\}$, suppose $S_1$ seeks to collaborate with $S_2$. If the condition for merge is satisfied, a new coalition $S_1:=S_1 \cup S_2$ is formed, otherwise, $S_1:=S_1$ and $S_1$ attempts to merge with another coalition who shares a mutual benefit in merging. The algorithm is repeated for the
remaining $S_i$ until all the coalitions have made their
merge decisions. The resulting partition $\mathcal{R}$ is then subject to a split process in a similar way. Then, successive merge-and-split processes go on until the iterations terminate.

The stability of this resulting network structure can be
investigated using the concept of a defection function $\mathbb{D}$~\cite{Apt_ageneric, Zhu2009}.
\begin{mydef}
A defection function $\mathbb{D}$ is a function which is associated with each partition $\mathcal{T}$. A partition $\mathcal{T} = \{T_1, \dots, T_m\}$ is $\mathbb{D}$-stable if no group of players is interested in leaving $\mathcal{T}$ when the players who leave can only form the coalition allowed by $\mathbb{D}(\mathcal{T})$. 
\end{mydef}

A partition $\mathcal{T}= \{T_1, \dots, T_m\}$ of $\mathcal{N}$ is $\mathbb{D}_{hp}$-stable, if no players in $\mathcal{T}$ are interested in leaving $\mathcal{T}$ through \textit{merge-and-split} to form other partitions in $\mathcal{N}$. A partition $\mathcal{T}$ is $\mathbb{D}_c$-stable, if no players in $\mathcal{T}$ are interested in leaving $\mathcal{T}$ through \textit{any operation} to form other collections in $\mathcal{N}$~\cite{Zhu2009}. 

$\mathbb{D}_{hp}$-stable can be thought of as a state of equilibrium where no coalitions have an incentive to pursue coalition formation through merge or split. The following theorem has been proved in~\cite{DBLP:journals/corr/abs-cs-0605132}. 
\begin{theorem}
A partition is $\mathbb{D}_{hp}$-stable if and only if it is the outcome of iterating the merge-and-split rules. 
\end{theorem}
\begin{remark}
For the proposed $(\mathcal{N}, V)$ collaborative distributed inference game, the proposed merge-and-split algorithm converges to a $\mathbb{D}_{hp}$-stable partition. 
\end{remark}
It is known that if $\mathcal{T}$ is $\mathbb{D}_c$-stable, then $\mathcal{T}$ is the outcome of every iteration of the merge-and-split rules and it is a unique $\mathbb{D}_c$-stable partition~\cite{Apt_ageneric}. 
Nonetheless, a $\mathbb{D}_c$-stable partition does not always exist. A  $\mathbb{D}_c$-stable partition is not guaranteed for our collaborative game and its existence depends on the specific characteristics of the sensor network and the parameters of the cost function in \eqref{ccf}.

\begin{remark}
For the proposed $(\mathcal{N}, V)$ collaborative distributed inference game, the proposed merge-and-split algorithm converges to the optimal $\mathbb{D}_c$-stable partition, if such a partition exists.
Otherwise, the proposed algorithm converges to a $\mathbb{D}_{hp}$-stable partition.
\end{remark}

\begin{proof}
By the properties of $\mathbb{D}_c$-stable partition shown in~\cite{DBLP:journals/corr/abs-cs-0605132}, $\mathbb{D}_c$-stable partition is a unique outcome of any arbitrary merge-and-split iteration. Thus, if a $\mathbb{D}_c$-stable partition exists, the merge-and-split algorithm finally converges to it~\cite{Walid2009}.
\end{proof}

\section{simulation results}
\label{sec:sim}
In this section, we present the simulation results of our proposed game theoretical approach to the collaborative distributed inference problem.  We consider a wireless sensor network with $N$ sensors deployed in a $[0,1.5] \times [0,1.5]$ square area of interest. Let the location of sensor $n$ be denoted by $\mathbf{s_n}=[s_{n1}, s_{n2}]$. The amount of dependence measured in terms of Kendall's $\tau$ between any two sensors $n$ and $m$ follows the power exponential model \cite{Vuran04spatio-temporalcorrelation}
\begin{eqnarray}
\tau(d_{n,m}) = \mathrm{e}^{-d_{n,m}^2},
\label{eq:depmodel}
\end{eqnarray} 
where $d_{n,m}=\| \mathbf{s}_n-\mathbf{s}_m \|$ is the distance between nodes $n$ and $m$ respectively located at coordinates $\mathbf{s}_n$ and $\mathbf{s}_m$.

We first consider a 8-sensor network where each sensor's observation follows Gaussian distribution with mean $\theta$ and variance $\sigma_n^2$, and the inter-sensor dependence is described by a Gaussian copula. Let $\mathbf{s_S}$ denote the location of the signal source which is $[0.75, 0.75]$ in this experiment. The variance of each sensor's observation is inversely propotional to the distance between the sensor and the signal source, i.e., $\sigma_n^2 = 1/|\mathbf{s_n}-\mathbf{s_S}|$. We set $rE_t = 1$ and  $\alpha = 4$, thus, according to (\ref{eq:energycon}), the largest coalition size that satisfies the energy efficiency constraint is $|S|=4$.  

In the problem of estimation, the prior distribution of the unknown parameter $\theta$ is assumed to have standard Gaussian distribution (zero mean, unit variance). 
The average FI of coalition $S$ is given as 
\begin{eqnarray}
I(S) = \mathbf{1}^T \Sigma_{S}^{-1} \mathbf{1}
\end{eqnarray}
where $\mathbf{1}$ is an all one vector with dimension $|S|$ by $1$, and $\Sigma_{S}$ is the covariance matrix of coalition $S$, i.e., 
$
\Sigma_{S} = \left[ \sigma_{m,n}\right]_{m,n \in S}
$
with $\sigma_{m,n}$ representing the covariance of sensor $m$ and sensor $n$.
In the detection problem, we set the parameters under hypothesis $H_0$ and $H_1$ to be $\theta_0=0$ and  $\theta_1=\sqrt{2}$.  
The KLD corresponding to a coalition $S$ is 
\begin{eqnarray}
D(S) = \mathbf{1}^T \Sigma_{S}^{-1} \mathbf{1}
\end{eqnarray}
With the above setting, the average FI and KLD have exactly the same expression. Thus, we present the simulation results without distinguishing between the problems of estimation and detection. 

%

In the initialization step, each sensor is set to be a coalition by itself, i.e., $\mathcal{S}=\left\{\{1\}, \{2\}, \{3\}, \{4\}, \{5\}, \{6\}, \{7\}, \{8\}\right\}$. By applying the proposed merge-and-split algorithm iteratively, three coalitions are formed as shown in Figure \ref{fig:network}. It can be seen that each coalition contains physically apart, and thus statistically less dependent, sensors so that redundancy loss is avoided and diversity gain is taken advantage of to the largest degree. Also, the sensors closer to the signal source, who already have a good individual performance, form smaller coalitions, while the distantly located sensors form relatively large coalitions to improve their performance. Since a $\mathbb{D}_c$-stable solution is not guaranteed in this example, the resulting partition of the merge-and-split algorithm may change with different initializations. With each iteration of merge-and-split, the overall payoff \footnote{We use ``overall payoff" to imply the payoff averaged over all sensors. The term ``overall" will continually  be used  with the same implication in the later part of this section.} of sensors increases, until no further merge or split occurs as shown in Figure \ref{fig:avgpayoff}.

\begin{figure}[b!]
\centering
\includegraphics[width=0.5\columnwidth]{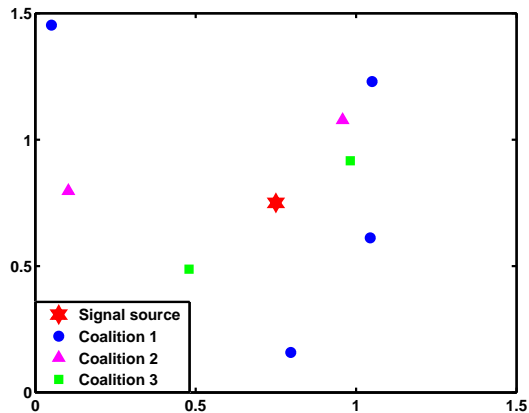}
\caption{The deployment of the 8-sensor network and the final partition.}
\label{fig:network}
\end{figure}

\begin{figure}[b!]
\centering
\includegraphics[width=0.5\columnwidth]{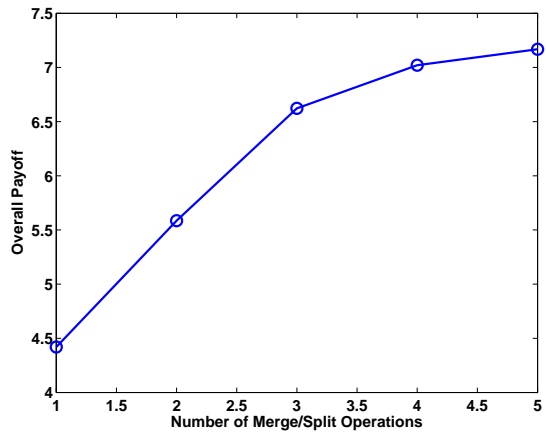}
\caption{Average inference performance increases with each merge or split operation.}
\label{fig:avgpayoff}
\end{figure}

We further consider a heterogeneous sensor network consisting of $28$ sensors deployed in the same area of interest. We assume that observations of 14 sensors follow Gaussian distribution with $\theta$ being the mean and unit variance, while observations of the other 14 sensors follow exponential distribution parameterized with $\theta$.  Within each Monte Carlo trial, the sensor locations are generated independently according to uniform distribution, through which the correlation matrix is obtained according to the dependence model in (\ref{eq:depmodel}). A student's t copula parameterized by the correlation matrix with the degree of freedom $\nu = 4$ is used to generate the dependence among sensors. 
The performance corresponding to different coalition formation approaches for this particular sensor deployment is evaluated. A total of 100 Monte Carlo trials are conducted and the performance is averaged over these trials.
We compare our proposed distributed algorithm based on coalition formation game with the approach of random coalition formation. 
In the random coalition formation method, a partition is randomly selected from the set of all partitions that satisfy the communication constraint with equality \footnote{The equality is to ensure a maximized inference performance, since the inference performance is nondecreasing in coalition size, according to Proposition \ref{prop:nonde} and Proposition \ref{prop:kldnonde}. We make the coalition size to be exactly $\alpha$, except for the one that may include less than $\alpha$ due to the fact that the total number of sensors $N$ may not be an integer multiple of $\alpha$. }.

\begin{figure}[b!]
\centering
\includegraphics[width=0.5\columnwidth]{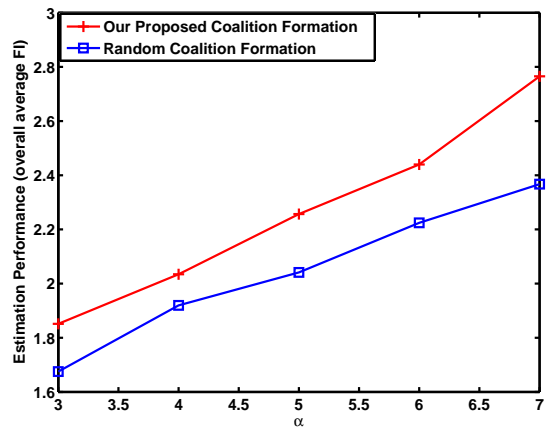}
\caption{Overall estimation performance vs. communication constraint $\alpha$.}
\label{fig:avgp_vs_alpha}
\end{figure}

\begin{figure}[b!]
\centering
\includegraphics[width=0.5\columnwidth]{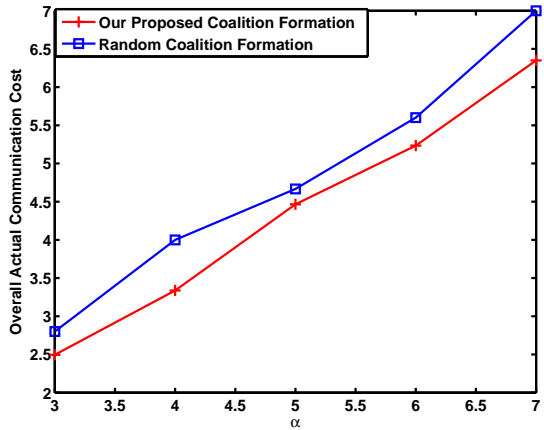}
\caption{Overall actual communication cost vs. communication constraint $\alpha$ in estimation problem.}
\label{fig:avgc_vs_alpha}
\end{figure}

In the estimation problem, $\theta$ is assumed to be standard Gaussian distributed. 
Figure \ref{fig:avgp_vs_alpha} shows the overall estimation performance of our proposed distributed coalition formation approach, compared with the random coalition formation approach. As the constraint on communication cost gets looser ($\alpha$ increases), the overall estimation performance becomes better for both methods. However, since our approach fully explores and utilizes inter-sensor dependence during the coalition formation process, it achieves much better performance than the random selection method. The overall communication costs, defined as $1/|\mathcal{S}|\sum_{S \in \mathcal{S}} E(S)$, are plotted in Figure \ref{fig:avgc_vs_alpha}, which demonstrates the superiority of our method in terms of communication efficiency. It has to be noted that, the average communication cost corresponding to our distributed coalition formation method is not the maximum communication cost that is allowed by the predefined constraint. It reflects the true cost of communication of the resulting partition, which may be much less than the maximum cost allowed. 

In the detection problem, we set $\theta_0=1$ and $\theta_1=2.4$. Superior overall detection performance of the partitions resulting from our proposed coalition formation approach is shown in Figure \ref{fig:Det_KLD_alpha}, in comparison with the random coalition formation. The overall actual communication costs versus $\alpha$ are plotted in Figure  \ref{fig:Det_com_alpha}, demonstrating a better communication efficiency of our approach.

In our distributed algorithm, when two coalitions are unable to contribute much to each other in inference performance due to their dependence (or redundancy loss incurred), they will not merge into a new coalition. 
Thus, it forces the coalition to seek cooperation with other coalitions to which it can contribute more, or where it is highly valued due to the diverse information that it is able to bring in. In this way, the overall diversity gain is increased while redundancy loss is decreased. 
By formulating the distributed inference problem as a coalition formation game and solving the game using an iterative algorithm, we are able to obtain better system performance in terms of both inference performance and energy efficiency, compared with the random coalition formation scheme. 

In numerous practical scenarios, sensor networks are subject to changes. For example, sensors embedded in people's cellphones change locations frequently. New sensors joining or existing sensors quitting also contributes to the time varying nature of the network. The distributed nature of our proposed coalition formation method in which sensors form coalitions automatically, makes it suitable for networks with time-varying configurations. 

\begin{figure}[b!]
\centering
\includegraphics[width=0.5\columnwidth]{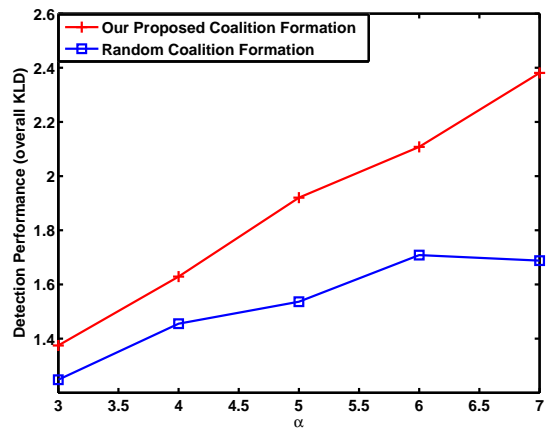}
\caption{Overall detection performance vs. communication constraint $\alpha$.}
\label{fig:Det_KLD_alpha}
\end{figure}

\begin{figure}[b!]
\centering
\includegraphics[width=0.5\columnwidth]{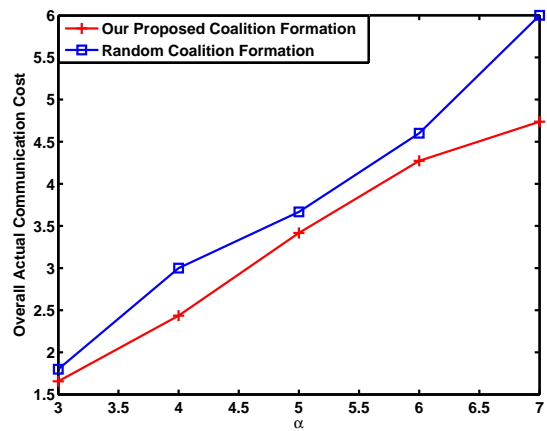}
\caption{Overall actual communication cost vs. communication constraint $\alpha$ in detection problem.}
\label{fig:Det_com_alpha}
\end{figure}

\section{conclusion}
\label{sec:con}
In this paper, we investigated a collaborative distributed inference problem in an energy constrained wireless sensor network with dependent observations. In the collaborative setting, sensors form coalitions and share observations within the coalition for improved inference performance. We focused on the formation of non-overlapping collaborating coalitions such that each sensor's performance is maximized while the energy constraint is satisfied. To analyze the benefit and cost of forming a certain coalition, we used copula theory to describe the dependence among observations, which provided ``redundancy'' and ``diversity'' aspects of inter-sensor dependence, respectively for the problem of estimation and detection. We defined GAFI and GKLD to quantify the diversity gain and redundancy loss in forming a coalition due to inter-sensor dependence. A coalition formation game was proposed for the distributed inference problem. A merge-and-split algorithm was utilized for our coalitional game and the stability of the outcome of our proposed algorithm was analyzed. Finally, numerical results were provided to demonstrate the performance of our game theoretical approach. Further investigation of the dependence-related concepts of diversity gain and redundancy loss in inference problems under different scenarios is to be conducted in the future work.

\section*{Acknowledgements}
Research was sponsored by ARO grant W911NF-14-1-0339. 

\bibliographystyle{IEEEtran}
\bibliography{IEEEabrv,refs-f12}

\end{document}